\documentclass{article}
\usepackage{graphicx} 
\usepackage{hyperref}

\usepackage{mdframed}

\newcommand{\C}{\mathrm{C}}
\newcommand{\PP}{\mathrm{P}}
\newcommand{\BP}{\mathrm{BPP}}

\newcommand{\ES}{\text{EXPSPACE}}

\usepackage{xcolor}

\newcommand{\cnd}{\mskip 1mu|\mskip 1mu}

\usepackage{amsmath, amssymb, amstext, amsthm,epsfig}
\usepackage{authblk}

\usepackage{tocloft}

\theoremstyle{plain}
\newtheorem{theorem}{Theorem}
\newtheorem{lemma}[theorem]{Lemma}
\newtheorem{fact}[theorem]{Fact}

\title{On the computational power of $\C$-random strings}
\author{Alexey Milovanov\footnote{email:\url{almas239@gmail.com}}} 
\affil{LASIGE, Faculdade de Ciências, Universidade de Lisboa}

\date{January 2025}

\begin{document}

\maketitle
\begin{abstract}

Denote by $H$ the halting problem. Let $R_U := \{ x \cnd \C_U(x) \ge |x| \}$, where $\C_U(x)$ represents the plain Kolmogorov complexity of $x$ under a universal decompressor $U$. We demonstrate the existence of a universal $U$ such that $H$ is solvable in polynomial time with access to the oracle $R_U$. This result resolves a problem posed by Eric Allender in~\cite{ea} regarding the computational power of Kolmogorov complexity-based oracles.

\end{abstract}
\section{Introduction}
At an informal level, the Kolmogorov complexity of string 
$x$ is defined as the minimal length of a program that outputs 
$x$ on the  empty input. This definition requires clarification: there are various types of complexity (plain, prefix, and others), but even when fixed, there remains a dependency on the choice of programming language or decompressor. We will give all the definitions needed for this paper in the next section; we also refer the reader to~\cite{LiVit,suv} for formal definitions and the main properties of Kolmogorov complexity.

The concept of Kolmogorov complexity allows us to define the notion of an individual random string. This can be done as follows: string $x$ is called random if its Kolmogorov complexity $\gtrsim |x|$, where $|x|$ is the length of $x$. Of course, the formal definition again requires clarifications: in addition to the aforementioned nuances with the definition of complexity, it is necessary to formally define what $\gtrsim$  means. Let us assume that we have clarified all the details. Define by $R$ the set of all random strings. The question arises: How powerful is the oracle $R$? For example, what languages belong to $\textbf{P}^R$?

This and other similar questions were investigated in~\cite{abk, abkmr, sh}.

In these papers some lower bounds for $\textbf{P}^R, \textbf{P/poly}^R$ and other complexity classes were proven. These results are robust in the following sense: the results are valid for all reasonable definitions of $R$, in particular, it does not matter what type of Kolmogorov complexity we consider. 

The situation is different for upper bounds. The results in~\cite{afg, crelm, m} show some limits on the computational power of $R$ only for \emph{prefix complexity}. 

For example, the results in~\cite{afg, crelm} together with the lower bound obtained in~\cite{sh} show that 
$$ \text{EXP}^{\text{NP}} \subseteq \bigcap_U \PP^{ R_{\text{K}_U}}
\subseteq \ES.$$
 
Here, the intersection is for all universal prefix-free decompressors $U$.

The lower and upper bounds are significantly different for the prefix complexity but at least we understand that the intersection does not contain undecidable languages. For the plain complexity even this is unknown. 
Let us introduce some notations.

Denote by $H$ the halting problem: 
$$H:= \{ x \cnd \text{ program } x \text{ halts on the empty input}\}.$$

Denote
$R_U: = \{ x \cnd  \C_U(x) \ge |x|\}$, where $\C_U(x)$ is the plain Kolmogorov complexity of $x$ under a universal decompressor $U$.

{\bf Open question \cite{ea}}

Is $H \notin \PP^{R_U}$ for some universal $U$?
\vspace{0.5em}

There is some evidence that the answer to this question should be affirmative. Kummer\cite{kum} showed that for each universal machine $U$, there is a time-bounded disjunctive truth-table reduction from $H$ to $R_U$: That is, there is a computable function that takes $x$ as input, and produces a list of strings,
with the property that $x \in H$ if and only if at least one of the strings is in $R_U$. However, it
was shown in \cite{abk} that, no matter what computable time bound t one picks, there is some
$U$ such that the disjunctive truth-table reduction from $H$ to $R_U$
requires more time than $t$.
It should be noted that the analogue of Kummer's theorem does not hold for prefix complexity~\cite{mp}.

Let us consider another question. How powerful can  $R_U$ be for \emph{some} universal $U$? It was shown in~\cite{abk} that there exists a universal $U$ such that $H$ can be reduced to $R_U$ by disjunctive truth-table in double exponential time.

Our main result is the following
\begin{theorem}
\label{main}
There exists a universal decompressor $U$ such that $H \in \PP^{R_U}$.
\end{theorem}

At the simplest level, the idea of the proof can be explained as follows. 

The universal decompressor is  optimal up to an additive constant, which allows us to construct  certain exotic decompressors with specific properties. For example, there exists a decompressor $U_{\text{even}}$ such that for every $x$ the Kolmogorov complexity $\C_{U_{\text{even}}} (x)$ is always an even number (Lemma~\ref{ev}). We will also construct some exotic decompressor that allows to solve the halting problem.


To explain the idea, we will first prove a simpler result: there exists a universal decompressor $U$ such that $H$ can be solved by the \emph{ oracle function}  $F_U$; that is, on input $x$ the function $F_U$ outputs its Kolmogorov complexity $\C_U(x)$. 

Let $U_{\text{even}}$ be the decompressor mentioned above, where every minimal description has even length.


The decompressor $U$ is defined as follows for arbitrary $d \in \{0, 1\}^{*}$:

\begin{itemize}
    \item if $d = 00d'$ then $U(d):= U_{\text{even}}(d')$.
    \item if $d =1d'$ then $U(d):= U_{\text{even}}(d') = x$ if $x \in H$. 
\end{itemize}
Note that $U$ is computable because $H$ is enumerable. 
We claim that $U$ is universal and $x \in H$ iff  $\C_U(x)$ is odd. Indeed, since $U_{\text{even}}$ is universal then $U$ is also universal (we need only descriptions that starts with $00$). 
If $x \in H$, then the minimal description of $x$ under $U$ starts with $1$ and hence has odd length. Conversely, if $\C_U(x)$ is odd, then $x \in H$.


We, however, need to solve $H$ using a weaker oracle that only allows to 
 distinguish random strings from 
 not-random. The challenge arises when the string $x$ is non-random, as the oracle $R_U$ distinguishes only between random and non-random strings.
 Now we describe some ideas that allow to get around the problem. 
 

One useful tool is the result from~\cite{abkmr} stating that, for any universal decompressor $U$, the complexity classes $\BP^{R_U}$ and $\PP^{R_U}$ are equal.


Another tool is Lemma~\ref{hash}, which we state informally here: Let $x$ be a string of length $n$ with complexity $k$, and let $A$ be a random binary matrix of dimensions $k \times n$. Then, with high probability, the string $A \circ Ax$ (where ``$\circ$'' denotes concatenation) is algorithmically random, i.e., its length is close to its Kolmogorov complexity.


This observation allows to construct a universal decompressor $U$ such that $H \in \PP^{R_U}$. We use a similar approach as for oracle function but instead of string $x$   we directly look at  $A \circ Ax$. This string  can be considered as a ``finger print'' of $x$ (note that matrices are the family of $2$-universal hash-functions) and also this string is close to random by the previously-mentioned properties of  Lemma~\ref{hash}. We construct $U$ in such a way that $A \circ Ax$ is random iff $x \in H$.    

The remainder of this paper is organized as follows. In the next section we recall basic definitions and properties of Kolmogorov complexity. 

Then in Section~\ref{sec_prep} we prove some technical lemmas that are needed to define the universal decompressor $U$ for Theorem~\ref{main} (Section~\ref{sec_un}). Then we describe a polynomial-time algorithm that solves $H$ by using oracle $R_U$ (Section~\ref{sec_alg}) and prove its correctness (Section~\ref{sec_cor}).

\section{Definitions and some properties of Kolmogorov complexity}
Let us recall the definition of plain conditional and unconditional Kolmogorov complexity.

Let $U$ be an algorithm whose inputs and outputs are
binary strings. We will call such $U$ \emph{decompressor} and  define the (unconditional) complexity
$\C_U(x)$ of a binary string $x$ with respect to $U$ as follows:
$$ \C_U(x): = \min \{ |y|: U(y) = x \} $$
Any $y$ such that $U(y)=x$ we call a description (or $U$-description) of $x$.

Therefore, the complexity
of $x$ is defined as the length of the shortest (or minimal) description  of $x$. 

A decompressor $U$ is called \emph{universal} if for every decompressor $U'$ there exists a constant $M$ such that for every string $x$ it holds that
$$ \C_U(x) \le \C_{U'}(x) + M.$$

Now we recall the definition of the \emph{conditional} Kolmogorov complexity. 
Let $D(p, x)$ be a computable partial function (that we will also call decompressor) of two string arguments; its values are also binary strings. We may think
of $U$ as an interpreter of some programming language. The first argument $p$ is considered as a
program and the second argument is an input for this program. Then we define the complexity function
$$\C_D (x \cnd y):=  \min \{|p|: D(p, y) = x\};$$
here $|p|$ stands for the length of a binary string $p$, so the right hand side is the minimal length of
a program that produces output $x$ given input $y$.

\begin{theorem}[Kolmogorov-Solomonoff]
 There exists a universal decompressor $D$ such
that for every other decompressor $D'$
there exists a constant
$M$ such that
$$\C_D(x \cnd y) \le \C_{D'} (x \cnd y) +M$$
for all strings $x$ and $y$.    
\end{theorem}
The proofs for all statements of this section can be found in~\cite{suv}.

Every decompressor $U$ for the conditional  complexity defines unconditional complexity can be naturally derived from the conditioned version by considering the condition as the empty string: 
$$\C_U(x):= \C_U(x \cnd \text{ the empty string}).$$

It is easy to see that if $U$ is a universal decompressor for the conditional Kolmogorov complexity then the corresponding decompressor for the unconditional Kolmogorov complexity is also universal.  

We will use the following basic properties of Kolmogorov complexity, which hold for any universal decompressor $U$:
\begin{itemize}
    \item \textbf{Bounded by length:} There exists a constant $M'$ such that:
    $$\C_U(x) \le |x| + M'.$$
    \item \textbf{Computable transformations:} For every total computable function $f$, there exists a constant $M'$ such that:
    $$\C_U(f(x)) \le \C_U(x) + M'.$$
    \item \textbf{Counting descriptions:} For every string $y$ and every $k$, the number of strings $x$ satisfying $\C_U(x \cnd y) \le k$ is at most $2^{k+1}$.
\end{itemize}

The Kolmogorov complexity $\C_U(x, y)$ of a pair of strings $x$ and $y$ is defined as follows. Let $x, y \to [x, y]$ be an injective computable function that maps a pair of strings into a single string. Then:
$$\C_U(x, y) := \C_U([x, y]).$$
This definition depends on the choice of the pairing function $[,]$, but only up to an additive $O(1)$ term.
Similarly, the complexity of triples of strings, integer numbers or any finite object can be defined.

For every natural number $M$ its Kolmogorov complexity $\C_U (M)$ is not greater than $\log M + O(1)$ (as the length of its binary representation). 

We will use the following fact about the complexity of pairs.
\begin{theorem}
\label{pair}
    $$ \C_U(x, y) \le \C_U(x) + \C_U(y \cnd x) + 2 \log \C_U(y) + O(1).$$
\end{theorem}

Moreover, it is known that this inequality is actually an equality up to  $O(\log (\C_U(x) + \C_U(y)))$ term which is known as the  \emph{Symmetry of Information}. 

The next two statements are known but not so common. We will provide proofs for these statements for the reader's convenience.

\begin{lemma}
\label{l_cc}
For every string $x$ it holds that 
$$ \C_U(x, \C_U(x)) = \C_U(x) + O(1).$$
\end{lemma}
\begin{proof}
The inequality $$ \C_U(x, \C_U(x)) \ge \C_U(x) + O(1)$$
obviously holds: every description for the pair $(x, \C_U(x))$ provides a description for $x$ that has the same length up to an additive constant.  

On the other hand, the shortest description of $x$ (of length $\C_U(x)$) determines
both $x$ and   $\C_U(x)$.   
\end{proof}
\begin{lemma}
\label{l_cond}
Let $\C_U(y \cnd n) = s$. Then $$\C_U(y \cnd n,s) = s - O(1).$$     \end{lemma}
\begin{proof}
Assume that there is a description $q$ for $y$ given $n,s$ of length $s - M$ for some $M$. We claim  that there is a description for $y$ given $n$ of length $s - M + O(\log M)$.
Indeed, knowing $q$ and $M$ it is possible to find $s$ (as, $s = |q| + M$). So: $$\C_U(y \cnd n) \le 
\C_U (q,M) + O(1).$$ Now we estimate $$C_U (q,M) \le |q| + O(\log M) = s - M + O(\log M).$$

Since     $s - M + O(\log M)$ must be not smaller than $s$ we conclude that $M$ is constant.

\end{proof}

\section{Preliminary lemmas}
\label{sec_prep}
In this section we prove some technical lemmas that will be used in the next sections. We start with the following folklore fact.
\begin{lemma}
\label{ev}
There exists a universal decompressor $V$ for the \emph{conditional} Kolmogorov complexity such that for every string $x$ and $y$, the value $\C_V(x \cnd y)$ is even. 
\end{lemma}
\begin{proof}
Let $V_{\text{opt}}$ be an arbitrary universal decompressor for the conditional    Kolmogorov complexity. Define $V(d, y)$ for arbitrary $d$ and $y$ in the following way:
\begin{itemize}
    \item if $d = 0\circ d'$ for  odd $|d'|$ then $V(d, y):= V_{\text{opt}}(d', y)$;
    \item if $d = 11\circ d'$ for  even $|d'|$ then $V(d, y):= V_{\text{opt}}(d', y)$;
    \item Otherwise
$V(d, y)$ is undefined. 
\end{itemize} 
From this definition it follows that for arbitrary $x$ and $y$ the value $\C_V(x \cnd y)$ is even and it holds that:
$$\C_V(x \cnd y) \le \C_{V_{\text{opt}}}(x \cnd y)+2.$$
Therefore, $V$ is a universal decompressor.
\end{proof}
As in the previous section we define  $\C_V(x)$ for arbitrary $x$  as $\C_V(x \cnd \text{ empty word})$.

This decompressor $V$ will be instrumental in constructing a universal decompressor $U$ that enables solving the halting problem.

We now prove two lemmas and define two constants ($D$ and $G$), which will be used in the next section to define $U$.

We will use the following simple fact about matrices.
\begin{fact}
\label{f}
Let $b_1 \not= b_2$ be two binary strings of length $n$. Let $A$ be a random under the uniform distribution  binary $k \times n$ matrix. Then the probability of the event $[Ab_1 = A b_2 ]$ is equal to $2^{-k}$.     
\end{fact}

\begin{lemma}
\label{hash}
Let $y$ be a string of length $n$. Denote $ k = \C_V(y \cnd n) + 5$.
Assume that $n+1$ is a prime number and $k <  n$.

 Let $A$ be a random  binary  $k \times n$ matrix (under the uniform distribution).

Let $f$ be the string $Ay$. Then for some constant $D$ the following holds with probability at least $\frac{9}{10}$ : 
    $$\C_V(A  \circ f) \ge nk + k - D.$$ Here $A$  is viewed as a string of length $nk$ .
\end{lemma}
\begin{proof}
{\bf First} we {\bf claim} that with probability at least $\frac{15}{16} = 1 - \frac{1}{2^4}$ there is no $t \not= y$ such that $At  = f$ and $\C_V(t \cnd n) \le k - 5$. Indeed, the number of strings $t$ such that $\C_V(t \cnd n) \le k - 5$ is at most $2^{k-4}$. Now we just use  Fact~\ref{f}.

The {\bf second claim} is that for some constant $C_1$  with probability at least $\frac{99}{100}$ it holds that $$\C_V(A  \circ y) \ge nk + k - C_1.$$

Let $\C_V(A  \circ y) \le nk + k - R$ for some $R$ and for $1$ percent of all possible  $A$. We want to show that $R = O(1)$.
Note that $$\C_V(y \cnd n, k, R) \le k - R + O(1).$$ Indeed, the number of strings $h$ of length $n$ such that there are at least  $1$ percent of all matrices $A$ such that $\C_V(A  \circ y) \le nk + k - R$ is at most $2^{k - R + O(1)}$. We can enumerate all such strings given $n, k, R$ and describe $y$ as the corresponding number in this enumeration.  


By  Lemma~\ref{l_cond} $\C_V(y \cnd n, k) = k - O(1)$. By using the same technique as in the proof of Lemma~\ref{l_cond} we conclude that $R$ is constant. 

 Now we prove the statement of the lemma. Let $$\C_V(A  \circ f) = nk + k - D \text{ for some }D.$$ We need to prove that $D = O(1)$ with high probability.
 
For this, we estimate the complexity of tuple $(A  \circ f, n,k)$.

First, note that by Lemma~\ref{l_cc} we have 
$$\C_V(A  \circ f, nk + k - D ) = nk + k - D + O(1).$$
By Theorem~\ref{pair} we can estimate the complexity of the following triple:
$$ \C_V(A  \circ f, nk + k - D, D) \le nk + k - D + O(\log D).$$

 Knowing $(n+1)k - D$ and $D$ it is simple to find $n$ and $k$ since $n+1$ is a large prime. Hence, 
 $$\C_V(A  \circ f, n,k) \le nk + k - D + O(\log D).$$ 
With high probability (at least $ 1 - \frac{1}{16} - \frac{1}{100}> \frac{9}{10}$) both claims above hold,
i.e. $y$ does not have collisions among strings  $t$ such that $\C_V(t \cnd n) \le k-5$  and $$\C_V(A \circ y) \ge nk + k - O(1).$$
Therefore, it is possible to find $y$ by $(A  \circ f, n,k)$ as the unique preimage of $A$ among strings with $\C_V(y \cnd n) \le k-5$.

Hence,  $D$ is constant.

\end{proof}

Now, we state a simple lemma that is an easy consequence of the fact that strings $x$ and $x \circ 0^{9}$ (just appending $9$ zeros to $x$) have the complexity up to some additive constant. 

 We need this lemma to define constant $G$ that we will use in the next section.
\begin{lemma}
\label{const}
For every $D$ there exists $G$ such that the following holds.  

For every $x$ if $\C_V(x) \ge |x| - D$ then $\C_V(x \circ 0^{9}) \ge (|x| + 9) - G$.

\end{lemma}

\section{Definition of  Universal Decompressor $U$}
\label{sec_un}
Here we describe  universal decompressor $U$ that allows to solve the halting problem in polynomial time. We will use decompressor $V$ from the previous section.

This decompressor $U$ will have the following properties: 
\begin{itemize}
    \item  every string $A \circ f$ of length $nk + k$ as in Lemma~\ref{hash} such that $\C_V(A \circ f) \ge nk + k - D$ is $U$-random (i.e. $\C_U(A \circ f) \ge |A \circ f| = nk +k$).
\item Lemma~\ref{rnr} (stated below) holds.
\end{itemize}

Odd  numbers $m$ such that $m = p \cdot u$, where $u  < p$ and $p$ is a prime number we  call \emph{specific}. For specific number $m$ the number $p$ is called its \emph{large prime factor}.

We define $U(d)$ for an arbitrary $d \in \{0,1\}^*$ as follows (using constants $D$ and $G$ defined in the previous section):

\begin{enumerate}
    \item If 
$d = 1^{D-1} 0 d'$
 for some $d'$ then $U(d)$ runs $V(d')$. 
 
 Assume that $V(d')$ outputs some string $y$  (if $V(d')$ is undefined, then $U(d)$ is also undefined).

If $y$ does not have a form $y= y' 0^{9}$ where $|y'|$ is specific then $U(d) := V(d')$.

Otherwise, (if $y$ has this form) we compare $|y|$ with $|d'|$.

If $|y| > |d'| - G$ then $U(d) := V(d')$. 
 Otherwise $U(d)$ is undefined.
 
   \item If $d = 1^{D }0 d'$  for some $d'$ then $U(d) := d'$. 
   \item If $d = 0d'$ for some $d'$ and $|d'| + 5$  is specific then $U$ works as follows. Let $p_l$---the $l$th prime number---be the large prime factor of $|d'|+5$. Let $|d'| = p_l \cdot k - 5$. Denote $n = p_l - 1$.
Consider $d'$ as the concatenation  $d' = A  \circ d''$, where $A$ is  the $k \times n$ matrix and 
$ d''$ is a  string of length $k-5$.

Run $V(d'' \cnd n)$. If it halts consider the output $y$.  If $|y| \not=n$ then $U$ is undefined in this case. Otherwise we consider $x$ --- the first $l$ bits of $y$.  If $x {\in} H$  then $U(d') := A \circ  Ay \circ 0^9$.
\end{enumerate}

Note that there are no strings with specific lengths that have $U$-description of the third type: all strings that have descriptions of the third type have even length: $$|A \circ Ay \circ 0^9| = nk + k + 9 = (n+1)k + 9 = p_l \cdot k + 9.$$ Since $k-5$ is the length of some $V$-description we conclude that $k$ is odd. So,  $p_l \cdot k + 9$ is even.

\begin{lemma}
\label{univ}
    Decompressor $U$ is universal. 
\end{lemma}
\begin{proof}
To prove this lemma we need only the descriptions of the  first and second type in the definition of $U$.

Indeed, consider some string $y$. Let $d'$ be a minimal $V$-description of $y$. 

If $|y|>|d'| - G$ then string $1^{D-1}0d'$ is an  $U$-description of $y$ and therefore $$\C_U(y) \le \C_V(y) + D.$$
If $|y| \le |d'| - G$ then 
$1^D 0 y$ is an $U$-description 
and therefore 
$$\C_U(y) \le |y| + D +1 \le |d'| - G + D + 1 = \C_U(y)- G + D + 1.$$
So, in both cases 
$$\C_U(y) \le \C_V(y) + D +1.$$
Therefore, $U$ is universal since $V$ is universal. 
\end{proof}
\begin{lemma}
\label{rnr}
Let $b$ be some string such that $|b|$ is specific and $|b| = p_l \cdot k$, where  $p_l$ is a  large prime. Consider $b$ as the concatenation $b= A \circ c$, where $|c| = k$.

 Assume that $\C_U(b) \ge |b|$ but $\C_U(b \circ 0^9) < |b| + 9$.

Then there exists a string $y$ of length $n:=p_l - 1$ such that the first $l$ bits of $y$ describe a program $x$ that belongs to $H$.

Moreover,  $\C_V(y \cnd n) \le k - 5$ and
$Ay = c$.
\end{lemma}
\begin{proof}
We claim that
 the inequality $\C_U(b) \ge |b|$ implies that $\C_V (b) \ge |b| - D$. Indeed, otherwise $b$ has an $U$-description of length less then $b$ of the first type. 
 
 (Note that since $|b|$ is specific, then there is no $b'$ such that $b = b' \circ 0^9$ and $|b'|$ is specific because the lengths of $b$ and $b'$ can not be both odd. Hence, the exception about $G$ does not apply to $b$.)

 Then by Lemma~\ref{const} we have that $\C_V(b \circ 0^9) \ge |b| + 9- G$. Hence, there are no $V$-descriptions of the first type for $b \circ 0^9$.

 Also, the minimal $V$-description of $b \circ 0^9$ is not of the second type because this string has a rather short description . Hence, $b \circ 0^9$ has a description of the third type that yields the conclusion of the lemma.
\end{proof}
\section{Algorithm}
\label{sec_alg}
We claim that the following probabilistic polynomial-time algorithm with access to $R_U$ solves the halting problem:

Let $x$ be an input instance for the halting problem. Denote by $l$ the length of $x$.

Denote by $p_l$ the $l$th prime number and set $n:= p_l - 1$.

Finally, denote by $y$ the following string of length $n$: $y: = x \circ 0^{n - l }$.

Then, 
for every  $k = 1, \ldots,  n-1 $   and 
for  $m = n^{50}$ 
we consider $m$ random matrices $A_1, \ldots, A_m$ of size $k \times n$.

If  for some $k$ and for at least half of $i \in \{1, \ldots, m\}$  it holds that
\begin{itemize}
    \item $A_i  \circ A_iy  \in R_U$  but
    \item $A_i \circ A_i y  \circ 0^9 \notin R_U$; 
\end{itemize}
 
 then the algorithm outputs that the program $x$ halts; otherwise, it outputs that $x$ does not halt.

\section{Correctness}
\label{sec_cor}
\begin{proof}[Proof of Theorem~\ref{main} ]
By Lemma~\ref{univ} the decompressor $U$ is universal and as it was mentioned in the Introduction $\PP^{R_U} = \BP^{R_U}$. So, it is enough to show that the algorithm in the previous section  outputs the correct answer with high probability.

{\bf If $x$  halts.} 
 We need to prove that for some  $k$ and for random $A$  the event
 \begin{equation}
\label{event}
A  \circ A y \in R_U \text{ and } A  \circ A y \circ 0^9 \notin R_U.    
\end{equation}
 happens with high probability.

Set $k= \C_V(y \cnd n) + 5$. Then, by Lemma~\ref{hash}, with  probability at least $\frac{9}{10}$  it holds that $$\C_V(A  \circ f) \ge nk + k - D.$$ Thus, by the definition of $U$ the event
$A  \circ A y  \in R_U$ holds.

We claim that $A  \circ A y \circ 0^9 \notin R_U$ with probability $1$.

 Note that the first $l$ bits of  $y$ is a string $x$ that belongs to $H$.  
 Consider the following $U$-description of type $3$: the string $0 \circ A \circ d''$, where $d''$ is the minimal $V$-description of $y$ conditional to $n$. By the definition of $U$ this is a description of $A \circ Ay \circ 0^9$. The length of the description is $$1 + k \cdot n + k-5 $$ that is less than $$|A \circ Ay \circ 0^9| = k \cdot n + k + 9.$$ Thus, we conclude that $A \circ A y \circ 0^9 \notin R_U$, proving correctness in this case.

{\bf If $x$ does not halt.} 

We claim that, in this case, for every $k$ and random $A$  the probability of the event~\eqref{event} is small. 

 More precisely, we show that this probability  is not greater than $\frac{1}{16}$. 

This implies that the algorithm works correctly with high probability.
 
 Indeed, the algorithm fails only if event~\eqref{event} holds for at least half of the matrices $A_1, \ldots, A_m$.
 Since matrices are chosen independently and $m$ is a sufficiently  large polynomial Hoeffding's inequality implies that this probability is exponentially small in $n$.

(In more detail,
define independent random variables $I_1, I_2, \ldots, I_m$, where  
\[
I_j =
\begin{cases}  
1, & \text{if event~\eqref{event} occurs for matrix } A_j, \\  
0, & \text{otherwise}.  
\end{cases}  
\]  

Since the expectation of each $I_j$ is at most $\frac{1}{16}$, it follows that  
\[
\mathbb{E} \left[ \sum_{i=1}^{m} I_i \right] \leq \frac{m}{16}.
\]  

By Hoeffding’s inequality, the probability that  
\[
\sum_{i=1}^{m} I_i \geq \frac{m}{2}
\]  
is exponentially small. ) 

Thus, the algorithm works correctly for all $k$ with high probability.



It remains to prove the claim.  

To analyze when event~\eqref{event} can happen, assume that for some $k$ and $A$, it does occur.  
By Lemma~\ref{rnr}, this implies the existence of some $y' \neq y$ of length $n$ such that  
\begin{itemize}
    \item $\C(y' \cnd n) \leq k - 5$, and  
    \item $A y' = A y$.  
\end{itemize}

(Since $y$ must also satisfy additional conditions, we only focus on these two properties.)  

For a fixed $k$, there are at most $2^{k-4}$ strings $y'$ such that $\C(y' \cnd n) \leq k-5$.  
Thus, by Fact~\ref{f}, for a random matrix $A$, the probability of $A y' = A y$ for some such $y'$ is at most $\frac{1}{16}$.



\end{proof}

\section{Acknowledgments}
I would like to thank  Eric Allender for  the formulation of the question and to Bruno Bauwens for the useful discussions.

This work was funded by the European Union (ERC, HOFGA, 101041696). Views and opinions expressed are however those of the author(s) only and do not necessarily reflect those of the European Union or the European Research Council. Neither the European Union nor the granting authority can be held responsible for them.
It was also supported by FCT through the LASIGE Research Unit, ref.\ UIDB/00408/2025 and ref.\ UIDP/00408/2025, and by CMAFcIO, FCT Project UIDB/04561/2020, \url{https://doi.org/10.54499/UIDB/04561/2020}.

\end{document}